\documentclass{article}

\usepackage{arxiv}

\usepackage[utf8]{inputenc} 
\usepackage[T1]{fontenc}    
\usepackage{hyperref}       
\usepackage{url}            
\usepackage{booktabs}       
\usepackage{amsfonts}       
\usepackage{nicefrac}       
\usepackage{microtype}      
\usepackage{lipsum}
\usepackage{graphicx}
\usepackage{subcaption}
\usepackage{amsmath}        
\usepackage{amsthm}        
\usepackage{mathtools}     
\usepackage{tikz}          
\usepackage{mathrsfs}      
\usepackage{bm}            

\newtheorem{theorem}{Theorem}
\newtheorem{lemma}{Lemma}[section]
\newtheorem{example}{Example}

\graphicspath{ {./images/} }

\title{Variables Ordering Optimization in Boolean Characteristic Set Method Using Simulated Annealing and Machine Learning-based Time Prediction}

\author{
 Minzhong Luo \\
  School of Computer Science\\
  Southwest University of Science And Technology\\
  No.59,Middle Section of Qinglong Avenue,Mianyang,China\\
  \texttt{luominzhong@swust.edu.cn} \\
   \And
 Yudong Sun \\
  Independent\\
  Houston, TX, 77077\\
  \texttt{yskysun@gmail.com} \\
  \And
 Yin Long \\
  School of Computer Science\\
  Southwest University of Science And Technology\\
  No.59,Middle Section of Qinglong Avenue,Mianyang,China\\
  \texttt{yinnuoloong@163.com} \\
  \And
 Zhu Genhao \\
  School of Computer Science\\
  Southwest University of Science And Technology\\
  No.59,Middle Section of Qinglong Avenue,Mianyang,China\\
  \texttt{zhugenhao@swust.edu.cn} \\
}

\begin{document}
\maketitle
\begin{abstract}
Solving systems of Boolean equations is a fundamental task in symbolic computation and algebraic cryptanalysis, with wide-ranging applications in cryptography, coding theory, and formal verification. Among existing approaches, the Boolean Characteristic Set (BCS) method\cite{Huang2014BCS} has emerged as one of the most efficient algorithms for tackling such problems. However, its performance is highly sensitive to the ordering of variables, with solving times varying drastically under different orderings for fixed variable counts $n$ and equations size $m$. To address this challenge, this paper introduces a novel optimization framework that synergistically integrates machine learning (ML)-based time prediction with simulated annealing (SA) to efficiently identify high-performance variables orderings. We construct a dataset comprising variable frequency spectrum $X$ and corresponding BCS solving time $t$ for benchmark systems(e.g., $n=m=28$). Utilizing this data, we train an accurate ML predictor $f_t(X)$ to estimate solving time for any given variables ordering. For each target system, $f_t$ serves as the cost function within an SA algorithm, enabling rapid discovery of low-latency orderings that significantly expedite subsequent BCS execution. Extensive experiments demonstrate that our method substantially outperforms the standard BCS algorithm\cite{Huang2014BCS}, Gröbner basis method \cite{B2009PolyBoRi} and SAT solver\cite{CryptoMiniSAT2012}, particularly for larger-scale systems(e.g., $n=32$). Furthermore, we derive probabilistic time complexity bounds for the overall algorithm using stochastic process theory, establishing a quantitative relationship between predictor accuracy and expected solving complexity. This work provides both a practical acceleration tool for algebraic cryptanalysis and a theoretical foundation for ML-enhanced combinatorial optimization in symbolic computation.
\end{abstract}


\section{Introduction}

\subsection{Research Background}
Solving Boolean equations, defined over the ring $\mathbb{F}_2[x_1, x_2, \ldots, x_n]/\langle x_1^2 + x_1, \ldots, x_n^2 + x_n \rangle$, is a central problem with profound and significant implications in computer science, symbolic computation, and cryptanalysis \cite{Cox2004Ideals,MC2008BSAT}. Boolean equations $\mathcal{P} = \{f_1, f_2, \ldots, f_m\}$ consists of $m$ polynomials in $n$ binary variables $x_1, x_2, \ldots, x_n$, where the task is to determine all solutions in $\mathbb{F}_2^n$ that satisfy $f_i(x_1,\ldots,x_n) = 0$ for $1 \le i \le m$, where solutions must satisfy both polynomial equations and field constraints. Such systems frequently arise in algebraic attacks on block ciphers, stream cipher analysis, hardware verification, and combinatorial optimization problems\cite{Bard2007MQ2SAT,Chai2008Char}, the NP-hard nature of this problem \cite{Hart1982C} necessitates efficient computational methods, particularly as cryptographic systems increasingly rely on the computational intractability of solving structured Boolean equations \cite{B2007Algebraic,C2002Crypt}.

\begin{figure}
\centering
\includegraphics[width=178mm]{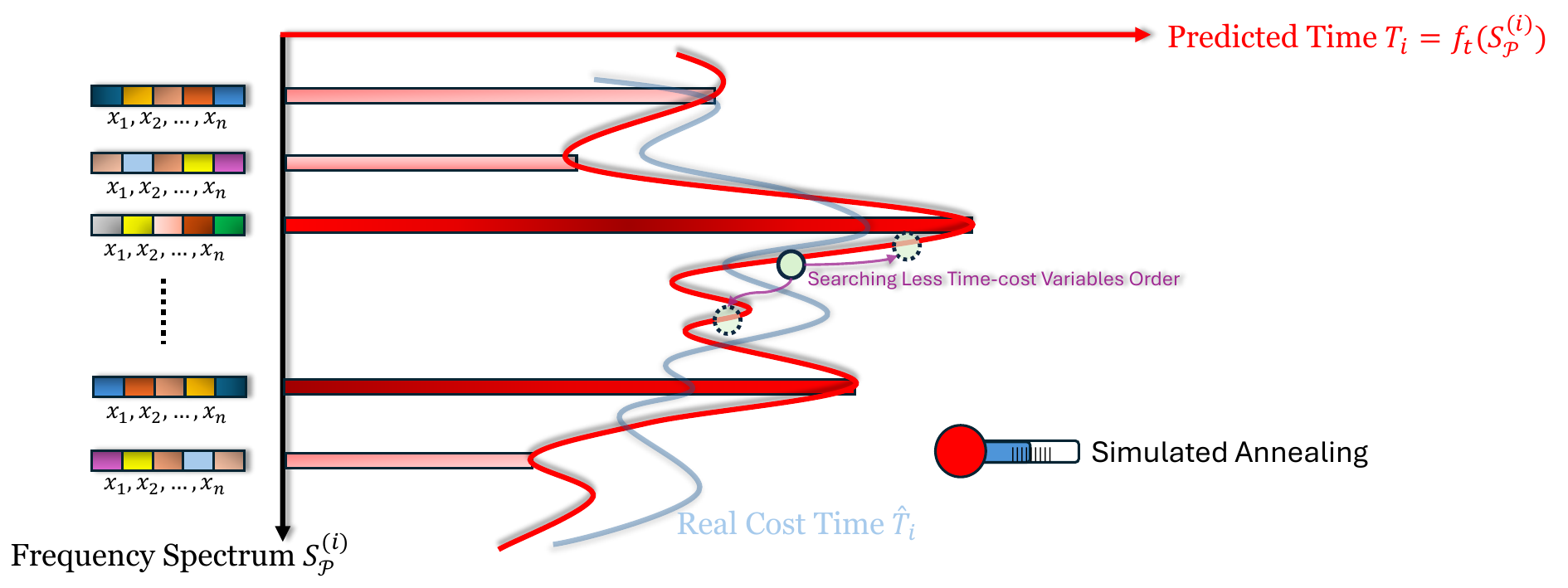}
\caption{Variables ordering optimization with simulated annealing and time predictor.}
\label{fig:Overview}
\end{figure}

Due to the inherent $\mathsf{NP}$-hardness of of solving Boolean equations, a wide variety of approaches have been developed from different perspectives, each exploiting distinct algorithmic mechanisms. In general, the representative methods can be categorized into the following three classes:

\begin{itemize}
    \item \textbf{Algebraic approaches}: These methods directly rely on tools from computational algebraic geometry and symbolic computation. A classical representative is the Gröbner basis method, which eliminates leading terms through polynomial reduction \cite{B2009PolyBoRi}. Despite its theoretical completeness, the Gröbner basis approach often suffers from exponential complexity as the number of variables and equations increases. Another significant line of research is Wu’s method, which has been adapted to the Boolean setting through triangular set decomposition. The resulting Boolean Characteristic Set (BCS) algorithm \cite{Huang2014BCS} provides a more efficient algebraic solver in practice and has become one of the most effective tools for Boolean system solving in cryptographic applications.
    
    \item \textbf{Transformation-based approaches}: Instead of directly solving the polynomial system, these methods transform the problem into an equivalent formulation that can be tackled by well-developed solvers in other domains. For example, Boolean equations can be translated into satisfiability problems and solved using SAT solvers \cite{Bard2007MQ2SAT}, or reformulated as mixed integer linear programming (MILP) problems and solved by MILP optimization techniques \cite{Abdel2010MILP}. Such transformations allow researchers to leverage the significant progress in SAT and MILP solver development. However, the conversion process may introduce additional variables or constraints, leading to a potential blow-up in problem size.
    
    \item \textbf{Numerical optimization-based approaches}: A different line of work formulates the Boolean system as a global optimization problem by designing an objective function whose global minimum corresponds to a solution of the original system. This allows the application of heuristic or metaheuristic optimization methods, such as simulated annealing \cite{Bor2010Hill}. While these approaches can be more flexible in handling large-scale or noisy problem instances, they typically do not guarantee completeness and may require careful tuning to achieve good performance.
\end{itemize}

In summary, existing methods for solving Boolean equations offer complementary advantages: algebraic approaches provide theoretical rigor, transformation-based approaches benefit from powerful external solvers, and optimization-based approaches enable heuristic scalability. Among the various algebraic techniques for solving polynomial systems, Wu's method and its Boolean variants\cite{Zhou2012App,Huang2011MFCS}, like the Boolean Characteristic Set (BCS) method\cite{Huang2014BCS}, have been recognized for their structured triangular decomposition of polynomial equations. Unlike Gröbner basis methods \cite{JC1999F4,JC2006F5}, the core idea of Wu's method is to transform the input system into a characteristic set of variables based on a chosen variables ordering \cite{Wang1993EM,Gao2011Char}, by pseudo-division and successive elimination. This process can significantly simplify the solving task by reducing multivariate problems into univariate subproblems, offering advantages for certain cryptographic applications \cite{Huang2014BCS}. However, a well-known challenge and critical limitation is that the computational complexity of the BCS solver is highly sensitive to the selected variables ordering, computational complexity varies dramatically across different variable permutations despite identical equation systems \cite{Chai2008Char}. Different permutations of the variables can lead to drastically different sizes of the zero-decomposition tree, resulting in solution time variations by several orders of magnitude. Empirical studies demonstrate solution times for fixed $(n,m)$ can span orders of magnitude under distinct orderings, rendering naive implementations impractical for industrial-scale problems.

The sensitivity to variables ordering from the elimination mechanism has motivated the design of heuristic and metaheuristic search strategies aimed at finding near-optimal orderings. Conventional approaches such as greedy selection or static heuristics often fail to scale efficiently with the problem size due to the factorial growth ($n!$) of the permutation space. Suboptimal orderings may generate intermediate polynomials with substantially higher degrees, exponentially increasing computational costs during triangular decomposition \cite{Huang2011MFCS,WM2017De}. Consequently, identifying high-quality orderings becomes essential for practical efficiency. Therefore, exploring intelligent optimization frameworks, especially those that can leverage predictive models or machine learning techniques to guide the search, has become a promising direction for enhancing the efficiency of Wu's method and its derivatives.

Recent advances propose machine learning to predict computational costs of mathematical procedures \cite{Tobias2023Conn}, including combinatorial optimizations in symbolic computation. However, integrating such predictors into global search frameworks for BCS variables ordering remains unexplored. This gap motivates our work combining ML-based time prediction with stochastic optimization to navigate the complex solution space efficiently.

\subsection{Our Results}
In this paper, we present a novel method for optimizing the variables ordering in the Boolean Characteristic Set (BCS) method using simulated annealing (SA) and machine learning (ML)-based time prediction. Our results are structured into three main contributions: correlation analysis, algorithmic performance, and theoretical insights.

We begin by performing a series of experiments to explore the relationship between various factors and the computational time of the BCS solver. One key finding is the relationship between the ascending and descending order of the variable frequency spectrum and the solving time. Interestingly, our analysis reveals that there is no direct correlation between the frequency spectrum order and the solving time. Specifically, while similar variable frequency spectrum do not necessarily correspond to similar solving times, and conversely, solving times that are close to each other may not align with similar frequency spectrum. This observation challenges the conventional wisdom that variables ordering based on frequency spectrum can directly predict solving efficiency, highlighting the need for more sophisticated optimization techniques.

Our proposed method, which integrates the SA algorithm with the ML-based time predictor, significantly outperforms the traditional BCS method. In our experiments, we show that the optimization framework is especially effective for larger problem sizes, such as those with variables $n=30$. When compared with the standard BCS algorithm\cite{Huang2014BCS}, Gröbner basis method \cite{B2009PolyBoRi} and SAT solver\cite{CryptoMiniSAT2012}, our approach consistently achieves a marked reduction in solving times. This improvement demonstrates the practical potential of combining machine learning and simulated annealing in solving Boolean equations more efficiently.

From a theoretical standpoint, we derive a probabilistic time complexity bound for our algorithm. In particular, we introduce a theorem that formalizes the relationship between the accuracy of the time predictor $f_t$ and the expected solving time of the BCS algorithm. The theorem asserts that a higher prediction accuracy leads to a more favorable expected solving time, thus providing a quantitative foundation for the practical efficacy of our approach. This relationship is crucial for guiding future improvements in predictor design and optimization strategies, which we will explore further in subsequent sections.

In summary, our results demonstrate that the proposed variables ordering optimization framework not only improves the performance of the BCS solver but also provides a solid theoretical foundation for integrating machine learning techniques with combinatorial optimization in symbolic computation.

\vspace{-1em} 
\nopagebreak 
\section{Preliminaries}

\subsection{Basic Terminology}
In this section, we introduce key concepts and notations that are foundational to understanding the optimization of variables ordering in the Boolean Characteristic Set (BCS) method and the role of simulated annealing (SA) in enhancing computational efficiency.

Boolean equations $\mathcal{P}$ defined over the ring $\mathbb{F}_2[x_1, \ldots, x_n]/\langle x_1^2 + x_1, \ldots, x_n^2 + x_n \rangle$, where $m$ polynomials $\mathcal{P} = \{f_1, f_2, \ldots, f_m\}$ in $n$ binary variables, the goal is to find all solutions $(x_1, \ldots, x_n) \in \mathbb{F}_2^n$ satisfying $f_i(x_1, \ldots, x_n) = 0$ for $1 \leq i \leq m$. This problem is NP-hard \cite{Hart1982C} and arises in areas like algebraic cryptanalysis, hardware verification, and combinatorial optimization \cite{Bard2007MQ2SAT, Chai2008Char}.

\begin{example}
To illustrate the structure of Boolean equations and how Wu’s method can be applied, consider the following system with five variables and three equations over $\mathbb{F}_2$:

$\begin{cases}
f_1 = x_1x_2 + x_3 + 1 = 0, \\
f_2 = x_2x_4 + x_5 = 0, \\
f_3 = x_1 + x_4 + x_5 = 0.
\end{cases}$

Here, each $f_i$ is a Boolean polynomial, and solutions are binary vectors $(x_1, x_2, x_3, x_4, x_5) \in \mathbb{F}_2^5$ that simultaneously satisfy all equations.  
Wu’s method proceeds by constructing a triangular set of polynomials through systematic elimination. For instance, from $f_2$ we obtain $x_5 = x_2x_4$, and substituting this into $f_3$ gives $x_1 = x_2x_4 + x_4$. Next, substituting $x_1$ into $f_1$ yields

$(x_2x_4 + x_4)x_2 + x_3 + 1 = x_2^2x_4 + x_2x_4 + x_3 + 1 = x_2x_4 + x_3 + 1 = 0$

where we have used $x_2^2 = x_2$ in $\mathbb{F}_2$. Thus, $x_3 = x_2x_4 + 1$.  
Combining these results, the solution set can be expressed as

$\begin{cases}
x_1 = x_2x_4 + x_4, \\
x_3 = x_2x_4 + 1, \\
x_5 = x_2x_4, \\
x_2, x_4 \in \mathbb{F}_2 \text{ are free variables}.
\end{cases}$

\end{example}

Hence, the system admits four solutions corresponding to the choices of $(x_2, x_4) \in \{0,1\}^2$. This simple example demonstrates how Wu’s method reduces the system into a triangular form, systematically expressing dependent variables in terms of a subset of free variables, thereby enabling the complete solution set to be enumerated efficiently.

Wu's method is a well-known approach for solving polynomial systems, transforms polynomial systems into a triangular form via pseudo-division and successive elimination under a chosen variables ordering $\sigma = (x_{i_1} \prec x_{i_2} \prec \cdots \prec x_{i_n})$ \cite{Wang1993EM, Gao2011Char}. The Boolean variant (BCS) \cite{Huang2014BCS} decomposes the system into a characteristic set $\mathcal{C}$, enabling efficient zero-solving through univariate subproblems. The primary advantage of the BCS method over Gröbner basis methods is its ability to reduce multivariate problems into simpler univariate problems, significantly improving computational efficiency. The effectiveness of BCS stems from:

1. Triangularization: Iterative elimination reduces multivariate equations to univariate polynomials, avoiding complex Gröbner basis computations \cite{JC1999F4}.  

2. Degree Control: Optimal variables orderings minimize intermediate polynomial degrees, curbing combinatorial explosion during decomposition \cite{Huang2011MFCS}.  
However, BCS is highly sensitive to $\sigma$; solution times vary exponentially across permutations due to divergent zero-decomposition tree structures \cite{Chai2008Char}.

The effectiveness of the BCS algorithm is closely tied to the order in which the variables are eliminated. Given a chosen variables ordering $\sigma = (x_{i_1} \prec x_{i_2} \prec \cdots \prec x_{i_n})$, the algorithm eliminates variables in the prescribed order, simplifying the system step by step. However, a key challenge is that the computational cost of solving the system is highly sensitive to the chosen variables ordering. Different permutations of variables can lead to drastically different solution times due to the varying complexity of the intermediate steps, especially the size and depth of the zero-decomposition tree. Thus, finding an optimal or near-optimal ordering is critical for improving efficiency.

Simulated annealing (SA) \cite{Van1987SA} is a probabilistic optimization technique inspired by the annealing process in metallurgy. In the context of optimizing the variables ordering for the BCS method, SA explores the large search space of possible variables orderings by iteratively perturbing a current solution (ordering) and accepting new solutions based on a probability function that depends on the difference in cost and a temperature parameter. The temperature is gradually reduced over time, which lowers the likelihood of accepting worse solutions, thus allowing the algorithm to converge to a near-optimal ordering.

SA effectively navigates the $n!$-sized solution space by escaping Local Minima with non-greedy acceptance of higher-cost solutions avoids premature convergence, and controlled exploration via temperature decay balances exploration (high $T$) and exploitation (low $T$), this stochastic nature allows SA to explore the complex solution space of variables orderings, where traditional deterministic methods may become trapped in suboptimal configurations, which ultimately leads to better global search performance as follows:

1. Initializes a random permutation $\sigma_0$.  

2. Generates neighbor $\sigma'$ by swapping two variables in $\sigma_i$.  

3. Accepts $\sigma'$ as $\sigma_{i+1}$ with probability $P = \min\left(1, e^{-\Delta E / T_i}\right)$, where $\Delta E = f_t(\sigma') - f_t(\sigma_i)$ is the cost change predicted by $f_t$, and $T_i$ is the temperature at step $i$.  

4. Cools $T_i$ geometrically: $T_{i+1} = \alpha T_i$ ($\alpha < 1$).  

In our method, SA is used in combination with a machine learning model that predicts the solving time for different variables orderings, with the goal of minimizing the expected solving time for the BCS method. By incorporating the predicted time as a cost function, the algorithm efficiently navigates the search space of variables orderings, resulting in significant reductions in computational time for solving Boolean equations.

\subsection{Sensitivity to Variables Order of BCS}
The Boolean Characteristic Set (BCS) algorithm is a Boolean variant of Wu's method for solving Boolean equations $\mathbf{P} = \{ f_1, f_2, \ldots, f_m \} \subseteq \mathbf{R}_2$ that decomposes a system of Boolean polynomials into a triangular form $\mathcal{A}_1, \ldots, \mathcal{A}_r$, where $\operatorname{Zero}(\mathbf{P}) = \bigsqcup_{i=1}^r \operatorname{Zero}(\mathcal{A}_i)$, enabling efficient solving through successive elimination \cite{Chai2008Char,Huang2014BCS,Gao2011Char}. Its effectiveness lies in the \emph{zero-decomposition mechanism}: during each elimination step, the algorithm selects a variable according to a prescribed order and performs pseudo-remainder reductions to split the problem into simpler subsystems with:

\begin{equation}
\operatorname{Zero}\left(I x_c+U\right)=\operatorname{Zero}\left(x_c+U, I+1\right) \cup \operatorname{Zero}(I, U)
\end{equation}

this recursive decomposition often produces univariate or nearly univariate polynomials, from which solutions can be computed efficiently. Compared with Gröbner basis methods \cite{JC1999F4,JC2006F5,B2009PolyBoRi}, the BCS method avoids excessive reductions to zero and provides a more direct decomposition strategy, making it particularly powerful for Boolean systems.

Although BCS is efficient in principle, its computational performance is strongly affected by the chosen variables order. To investigate this, we conducted experiments where the BCS solver was fixed and only the permutation of variables was changed. To quantify this, we fixed a Boolean system $\mathbf{P}$ and ran the BCS solver under different permutations $\sigma \in S_n$: for a given system, we generated multiple random variables orders $\sigma = (x_{i_1}, x_{i_2}, \ldots, x_{i_n})$ and recorded the solving time of BCS for each case. The statistical results show that the solving time may differ by several orders of magnitude across different orderings. In some instances, a carefully chosen order led to rapid decomposition and shallow zero-decomposition trees, while an unfavorable order caused deep branching and exponential blow-up in intermediate polynomial sizes. These results clearly demonstrate that the BCS method is highly sensitive to variables ordering.

Theoretically, this sensitivity originates from the algebraic elimination mechanism. At each step, the algorithm eliminates $\operatorname{lvar}(P)$ for some polynomial $P$, where the choice of $\operatorname{lvar}$ is dictated by $\sigma$. If $\operatorname{lvar}(P)$ has high interaction degree, i.e.,

\begin{equation}
\deg\big(\operatorname{lm}(P)\big) = d \quad \text{with many variables in its support},
\end{equation}

then substituting or decomposing on this variable generates new polynomials of larger total degree $\operatorname{tdeg}$ and higher density. This rapidly increases the size of the active polynomial set and deepens the decomposition tree. Conversely, if $\sigma$ postpones such variables and eliminates low-degree or linear variables first, then each decomposition reduces the system dimension while keeping $\operatorname{tdeg}$ small.  

Formally, let $\mathcal{T}(\sigma)$ denote the decomposition tree under order $\sigma$. Its size satisfies

\begin{equation}
|\mathcal{T}(\sigma)| \;\leq\; \prod_{k=1}^n b_k(\sigma),
\end{equation}

where $b_k(\sigma)$ is the branching factor when eliminating $x_{i_k}$. Since $b_k(\sigma)$ is strongly correlated with $\deg_{x_{i_k}}(\mathbf{P})$ and the structure of substitutions, different permutations $\sigma$ lead to exponentially different $|\mathcal{T}(\sigma)|$ and hence runtime $\tau(\sigma)$. This explains the large variance in solving time we observed experimentally.

In summary, the BCS algorithm is effective due to its recursive zero-decomposition and triangularization strategy. However, this very mechanism makes it highly sensitive to variables order: both experiments and theoretical analysis confirm that $\sigma$ dictates the complexity of intermediate polynomials and the size of the decomposition tree, which directly determines solving time.

\begin{figure}
\centering
\includegraphics[width=147mm]{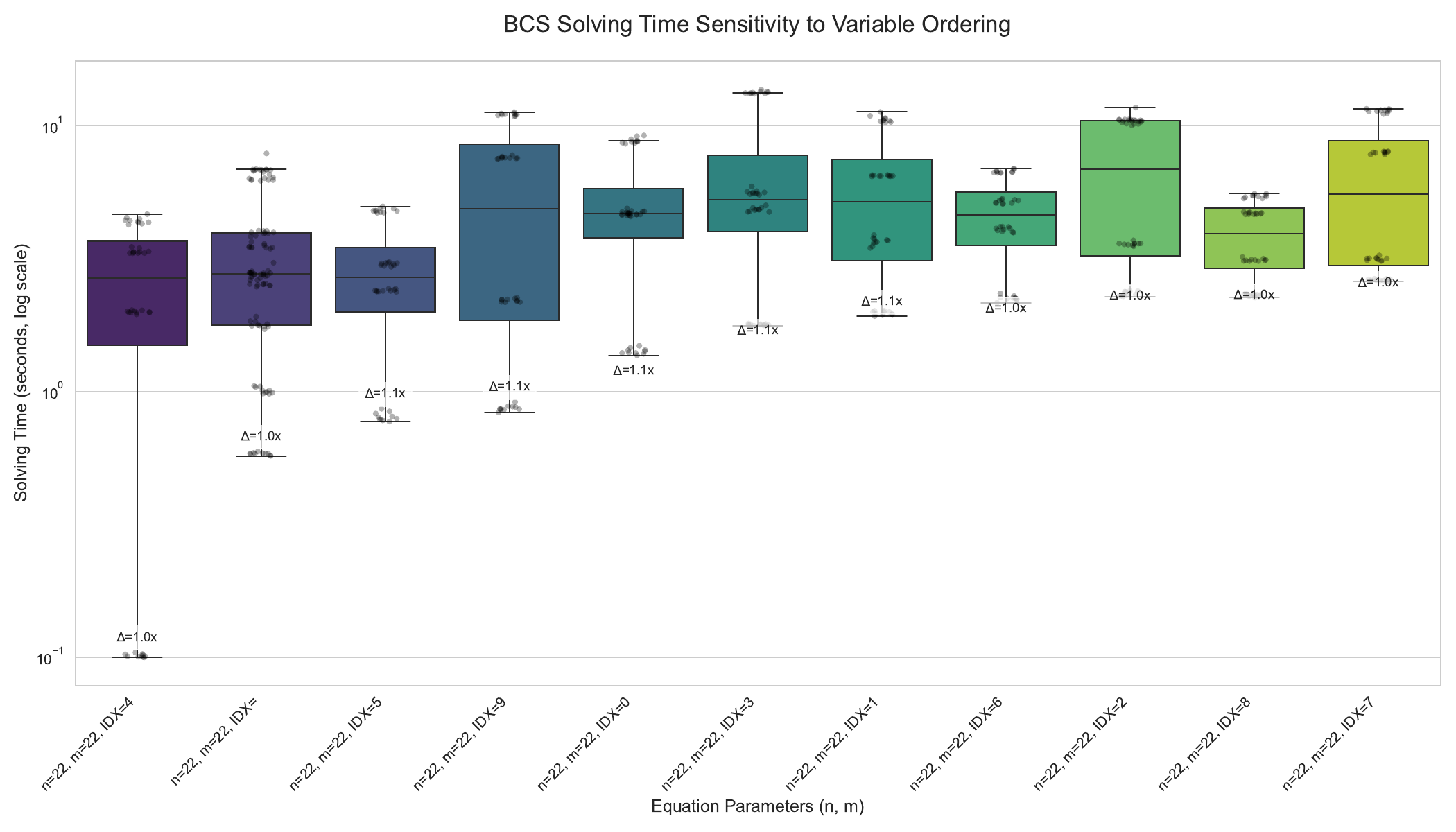}
\caption{Variables ordering impact on solving cost time.}
\label{fig:fig1}
\end{figure}

\section{Simulated Annealing with Time Predictor for Variables Ordering Optimization}
\label{sec:SAandTP}

\subsection{Correlation Analysis between Frequency Spectrum and Solving Performance}

We first conducted a basic experiment to verify the impact of variables ordering on the solving time of the BCS algorithm.  
Given Boolean equations $\mathcal{P} = \{f_1, f_2, \ldots, f_m\}$ in $n$ variables $x_1, \ldots, x_n$, we denote a variables ordering by a permutation $\sigma \in S_n$. For each $\sigma$, we apply the permutation to $\mathcal{P}$ (e.g., swapping $x_i$ and $x_j$ changes $f(x_i, x_j)=x_i x_j + x_i$ into $f(x_i, x_j)=x_i x_j + x_j$), and solve the permuted system with the BCS solver. The solving time is denoted as

\begin{equation}
\tau(\sigma) := \text{time}\big(\operatorname{BCS}(\mathcal{P};\sigma)\big).
\end{equation}

Our experiments confirm that $\tau(\sigma)$ exhibits substantial fluctuation across different $\sigma$, indicating that variables ordering directly influences computational cost.

To analyze the structural features of variables orderings, we define the \emph{frequency spectrum} of system $\mathcal{P}$ as

\begin{equation}
S_\mathcal{P} = (f_1, f_2, \ldots, f_n), \quad f_i := \frac{\#\{ \text{occurrences of } x_i \text{ in } \mathcal{P}\}}{\sum_{j=1}^n \#\{ \text{occurrences of } x_j \text{ in } \mathcal{P}\}},
\end{equation}

which can be regarded as a discrete probability distribution over variables. Intuitively, $S_\mathcal{P}$ captures the relative structural importance of each variable in the system.

We clustered spectrum $S_\mathcal{P}$ into $k=7$ clusters and compared the distribution of $\tau(\sigma)$ across them. The average silhouette coefficient was only $0.0926$, showing weak clustering effect. Moreover, the intra-cluster solving times showed large variance (e.g., cluster means ranged from $42.97$ to $60.73$ with standard deviations above $30$). This indicates that systems with similar spectrums do not necessarily yield similar solving times.

We also performed the inverse experiment: clustering systems by runtime values $\tau(\sigma)$, and analyzing whether their spectrum $S_\mathcal{P}$ are similar. For each runtime cluster, we computed the average intra-cluster spectrum correlation, obtaining values close to $0.03$-$0.05$. This demonstrates that similar solving times do not imply similar spectrums, further indicating that spectrum and runtime are not trivially related.

\begin{figure}[ht]
    \centering
    \begin{subfigure}[b]{0.45\textwidth}
        \includegraphics[width=\textwidth]{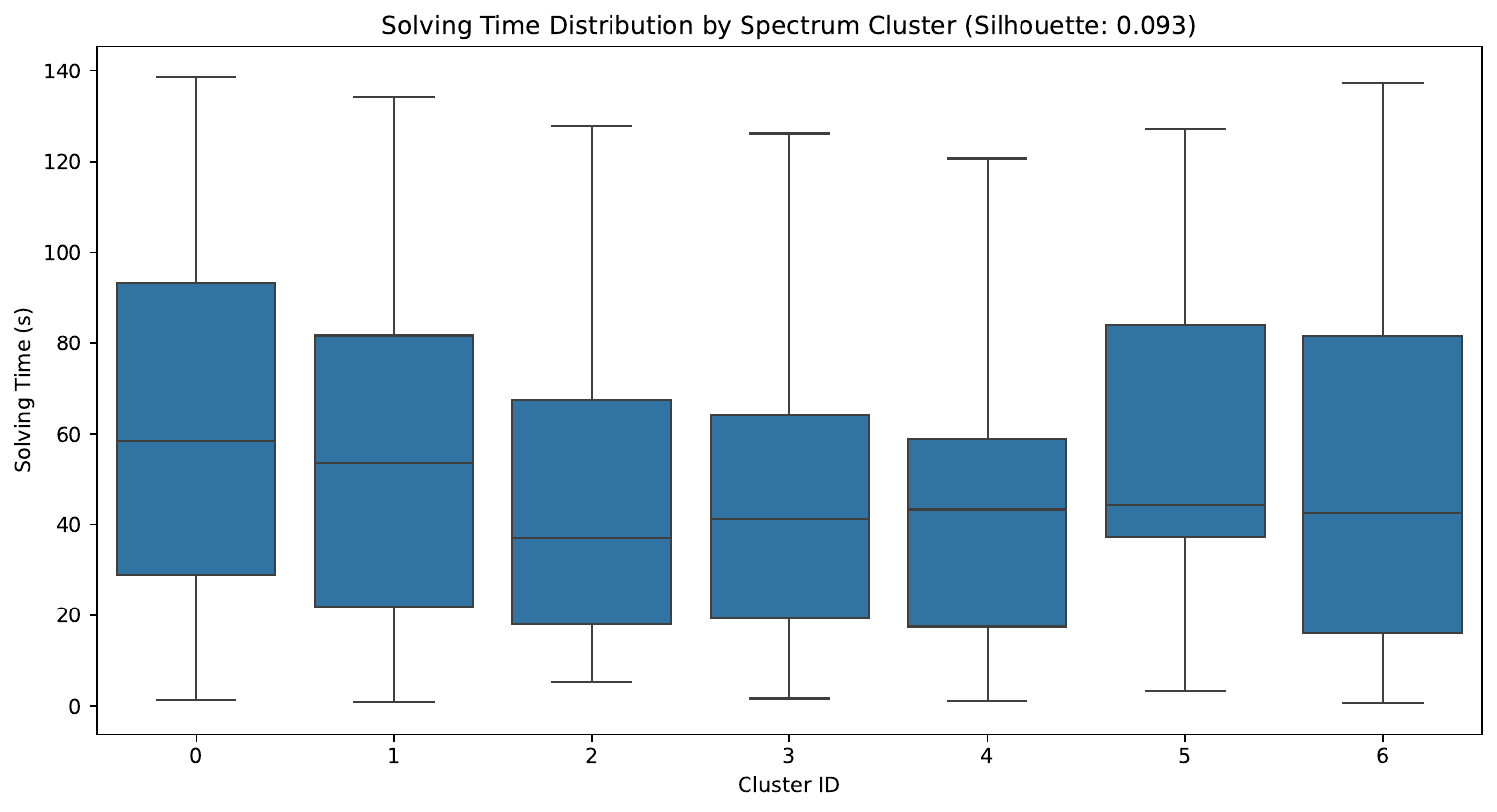}
        \caption{similar spectrums do not imply similar solving times;}
        \label{fig:sub1}
    \end{subfigure}
    \hfill
    \begin{subfigure}[b]{0.45\textwidth}
        \includegraphics[width=\textwidth]{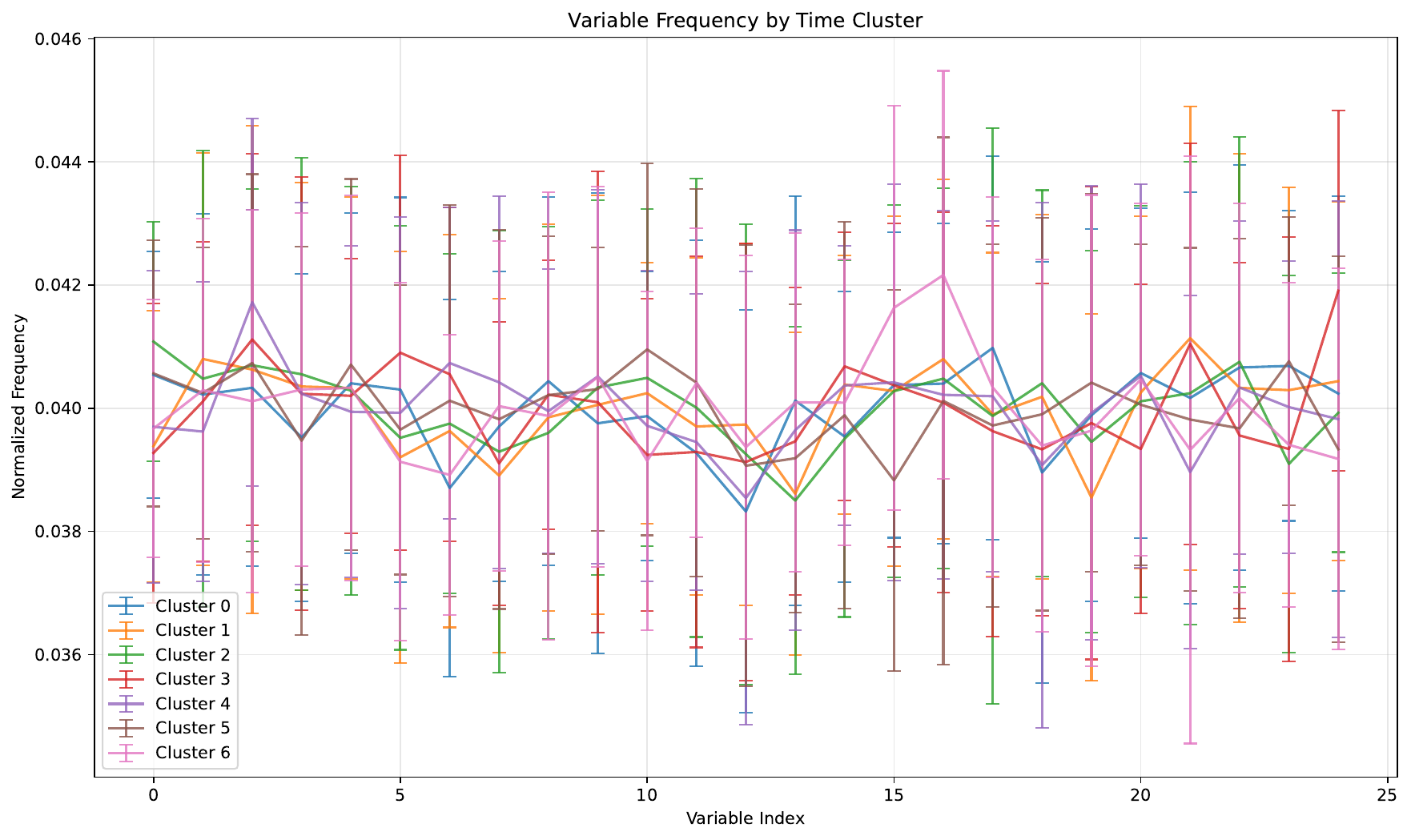}
        \caption{similar solving times do not imply similar spectrums;}
        \label{fig:sub2}
    \end{subfigure}
    \caption{correlation test between individual spectrum and runtime.}
    \label{fig:main}
\end{figure}

We next conducted a correlation test between individual spectrum components $f_i$ and runtime $\tau(\sigma)$. Several variables showed statistically significant but weak correlations, such as

\begin{equation}
\operatorname{corr}(f_0, \tau) = -0.27 \;(p<10^{-22}), \quad \operatorname{corr}(f_{14}, \tau) = 0.15 \;(p<10^{-7}), \quad \operatorname{corr}(f_{24}, \tau) = 0.14 \;(p<10^{-7}),
\end{equation}

while most other variables had negligible correlation. This suggests that frequency spectrum contains some predictive information about runtime, but the relation is not linear or directly interpretable.

From the above three experiments, we conclude that the variable frequency spectrum $S_\mathcal{P}$ and the solving time $\tau(\sigma)$ do not exhibit an obvious one-to-one correspondence:  
(1) similar spectrums do not imply similar solving times;  
(2) similar solving times do not imply similar spectrums;  
(3) some weak correlations exist, but not enough to form a deterministic rule.

Therefore, we hypothesize that the relationship between $S_\mathcal{P}$ and $\tau(\sigma)$ is mediated by complex latent factors, such as the structural interplay of variables in zero-decomposition trees. This motivates the construction of a predictor function

\begin{equation}
f_t : S_\mathcal{P} \mapsto \widehat{\tau},
\end{equation}

where $f_t$ learns the latent statistical relationship between spectrum and runtime. If such a predictor achieves nontrivial accuracy, it provides strong evidence that $S_\mathcal{P}$ and $\tau(\sigma)$ are statistically associated, and thus can guide intelligent optimization of variables ordering. This predictive modeling forms the basis of the next section.

\subsection{Implementation of the Solving Time Predictor}

Building on the preceding analysis of the statistical relationship between variable frequency spectra and solving time, in this section we implement a machine learning based predictor $f_t$ that maps the frequency spectrum of a Boolean system to an estimated solving time of the BCS algorithm. Formally, let $S_{\mathcal{P}} \in \mathbb{R}^n$ denote the normalized frequency spectrum of system $\mathcal{P}$, then the predictor aims to approximate the mapping

\begin{equation}
f_t : S_{\mathcal{P}} \mapsto \widehat{\tau}, \quad \text{where } \widehat{\tau} \approx \tau(\sigma),
\end{equation}

with $\tau(\sigma)$ denoting the actual solving time of BCS under ordering $\sigma$.

We compared several ensemble-based regression models, including Random Forest Regressor, AdaBoost Regressor, and Bagging Regressor, all of which are known to perform well on nonlinear regression tasks \cite{Song2019LearnVO,Tobias2023Conn}. Empirical evaluations showed that Gradient Boosting Regressor (GBR) achieved the most stable and accurate performance. 

The predictor is trained using pairs $\{ (S_{\mathcal{P}}, \tau(\sigma)) \}$ collected from benchmark different systems like $n=m=25$. The model is defined as

\begin{equation}
\widehat{\tau} = f_{\theta}(S_{\mathcal{P}}),
\end{equation}

where $f_{\theta}$ denotes the GBR with parameters $\theta = \{n_{\text{estimators}}, \text{learning\_rate}, \text{max\_depth}, \ldots \}$.  

We optimized hyperparameters to balance bias and variance: $n_{\text{estimators}}=1000$, learning rate $0.01$, maximum depth $20$, and leaf size $12$, with $80\%$ subsampling and $\sqrt{n}$ feature sampling. A $10\%$ validation split and early stopping ($n\_iter\_no\_change=50$) ensured generalization.

The fitting results are visualized in Figure~\ref{fig:32sub1} (prediction vs. ground-truth solving times) and Figure~\ref{fig:32sub2} (residual error distribution). Both figures show that the predictor captures the nonlinear mapping effectively:  
(1) the predicted solving times closely align with observed values;  
(2) residuals are concentrated around zero, with no severe systematic bias.

\begin{figure}[ht]
    \centering
    \begin{subfigure}[b]{0.45\textwidth}
        \includegraphics[width=\textwidth]{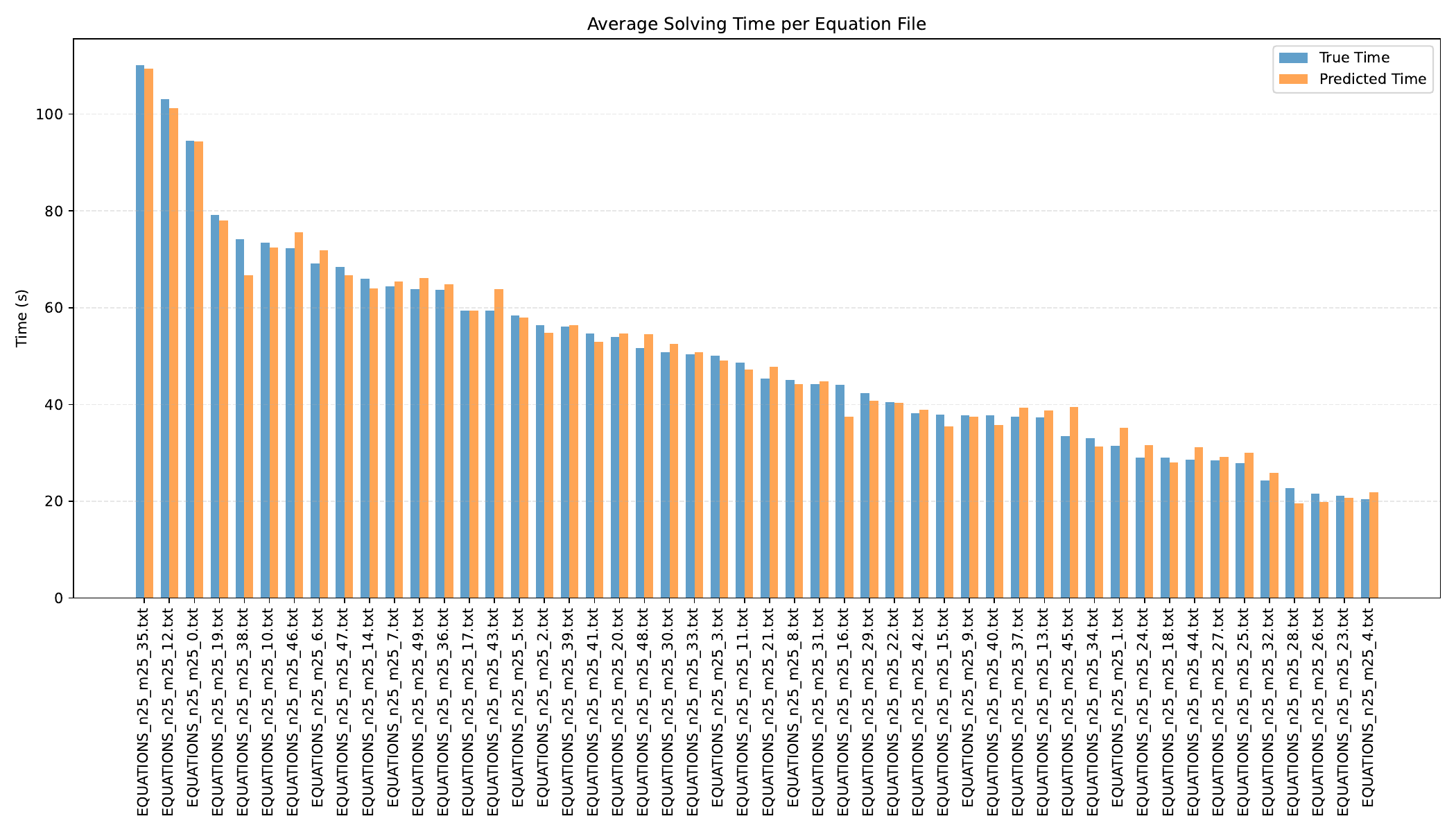}
        \caption{Comparison between predicted and actual solving time;}
        \label{fig:32sub1}
    \end{subfigure}
    \hfill
    \begin{subfigure}[b]{0.45\textwidth}
        \includegraphics[width=\textwidth]{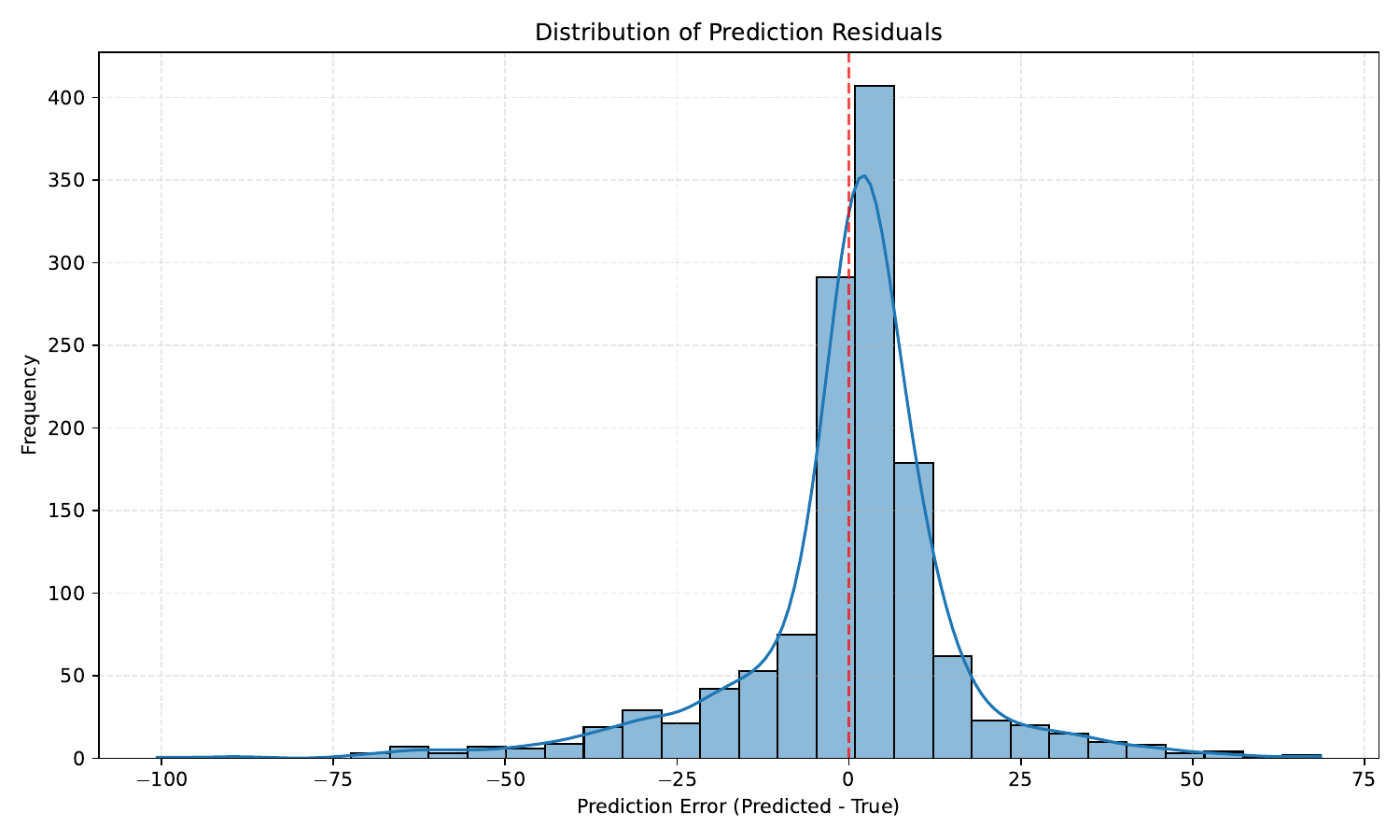}
        \caption{Residual distribution of solving time predictions;}
        \label{fig:32sub2}
    \end{subfigure}
    \caption{Fitting results of time predictor.}
    \label{fig:32main}
\end{figure}

Two conclusions can be drawn from this study:

\begin{enumerate}
    \item Although the clustering experiments in the previous section indicated that similar spectra do not necessarily correspond to similar solving times (and vice versa), this does not imply the absence of a relationship. Rather, the connection is hidden in a nonlinear and high-dimensional manner. The successful regression fitting demonstrates that $S_{\mathcal{P}}$ indeed contains informative signals, and that machine learning models can capture their latent associations with solving time.
    \item The obtained predictor $f_t$ provides sufficiently accurate solving time estimates, and thus can be integrated into our simulated annealing framework as the evaluation function. This enables efficient exploration of the permutation space, since costly BCS runs are replaced by lightweight predictor evaluations during the optimization stage.
\end{enumerate}

\subsection{Variables Ordering Optimization Based on Simulated Annealing}

Having constructed an accurate solving time predictor $f_t$, we now embed it into a simulated annealing (SA) framework to search for improved variables orderings of the BCS method. Recall that for a polynomial system $\mathcal{P}$ with variables $\{x_1, \ldots, x_n\}$, each ordering $\sigma \in S_n$ induces a distinct frequency spectrum $S_\sigma$ and a corresponding solving time $\tau(\sigma)$. Since $\tau(\sigma)$ varies drastically across $\sigma$ yet the search space $|S_n| = n!$ is prohibitively large, exhaustive search is infeasible. Instead, we model the search landscape as a "time field"(see Figure~\ref{fig:Overview})

\begin{equation}
\mathcal{F} = \{ (S_\sigma, f_t(S_\sigma)) : \sigma \in S_n \},
\end{equation}

where $f_t(S_\sigma)$ approximates the solving time at ``coordinate'' $S_\sigma$. Our objective is not necessarily to find the global optimum $\sigma^\ast$, but to locate a ``good enough'' ordering $\sigma'$ with significantly reduced solving time compared to the baseline ordering.

Classical simulated annealing \cite{Van1987SA} starts from an initial state $\sigma_0$, and iteratively explores its neighbors by applying random perturbations. At iteration $k$, a candidate ordering $\sigma_{k+1}$ is generated from $\sigma_k$ by a small modification (e.g., swapping two adjacent variables). The candidate is accepted with probability

\begin{equation}
P(\sigma_{k} \to \sigma_{k+1}) = 
\begin{cases}
1, & f_t(S_{\sigma_{k+1}}) < f_t(S_{\sigma_k}), \\
\exp\!\Big( - \dfrac{f_t(S_{\sigma_{k+1}}) - f_t(S_{\sigma_k})}{T_k} \Big), & \text{otherwise},
\end{cases}
\end{equation}

where $T_k$ is the temperature at iteration $k$, gradually decreasing with a cooling schedule.To adapt SA more effectively for variables ordering optimization in the BCS context, we introduce two tailored modifications:

\paragraph{1. Predictor-Guided Neighbor Generation.}
Instead of generating neighbors purely at random, we bias the perturbation mechanism using predictor feedback. Specifically, among all possible local swaps $(i,j)$, we pre-evaluate their effect using $f_t$, and prioritize those with larger predicted improvements. This guided neighbor generation accelerates convergence by avoiding low-potential moves, effectively embedding a local greedy heuristic within the stochastic search.

\paragraph{2. Adaptive Cooling with Predictor Confidence.}
The predictor $f_t$ provides not only a predicted time but also an implicit confidence derived from training residuals. We design an adaptive cooling scheme where the temperature decay depends on prediction stability: if the residual distribution in recent iterations is low (predictor is reliable), the temperature decreases faster; if residuals are high, temperature decays slower, allowing broader exploration. Formally,

\begin{equation}
T_{k+1} = \alpha \cdot T_k \cdot \big( 1 + \beta \cdot \mathrm{Var}(\epsilon_{k-h:k}) \big),
\end{equation}

where $\epsilon_{k-h:k}$ denotes recent residuals, $\alpha \in (0,1)$ controls cooling rate, and $\beta > 0$ adjusts sensitivity to predictor variance. This balances exploration and exploitation adaptively, reducing the risk of premature convergence.

The improved SA framework thus integrates machine learning prediction both as the objective evaluator and as a guide for search dynamics. Compared with standard SA, the two innovations---predictor-guided neighbor generation and confidence-adaptive cooling---are specifically tailored to the variables ordering optimization task. Together, they ensure that the algorithm (i) avoids wasting iterations on evidently poor neighbors, and (ii) adapts exploration depth to the reliability of predictor feedback. In the next chapter, we will experimentally evaluate this improved SA algorithm against baseline approaches (random order, greedy heuristics, and standard SA) to quantify its advantage in accelerating BCS solving for Boolean equations.

\begin{figure}
\centering
\includegraphics[width=147mm]{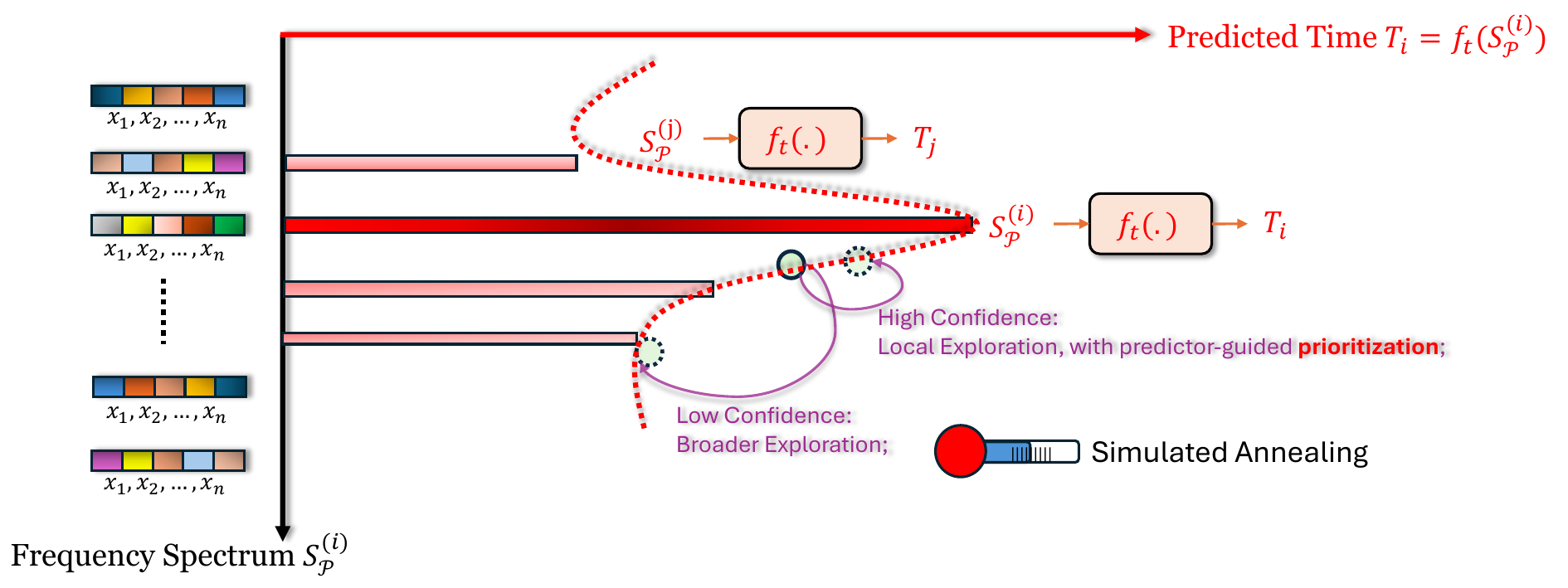}
\caption{he improved Simulated Annealing framework integrates machine learning prediction as a guide for searching process with confidence evaluated.}
\label{fig:fig1}
\end{figure}

\section{Experimental Results and Analysis}

\subsection{Experimental Results}
In this section, we evaluate the performance of our proposed algorithm for optimizing variables ordering in the Boolean Characteristic Set (BCS) method using simulated annealing (SA) and a time predictor. We compare our method with three baseline approaches: the Gröbner basis-based solver PolyBoRi, the SAT solver Cryptominisat, and the original BCS algorithm.

We selected a variety of Boolean equations, including both randomly generated instances and cryptographic systems, to ensure comprehensive evaluation. The following problem instances were considered:

\begin{itemize}
    \item Random Boolean system, $n=25$, $m=25$ (Square system).
    \item Random Boolean system, $n=25$, $m=50$ (Overdefined system).
    \item Random Boolean system, $n=28$, $m=28$ (Square system).
    \item Random Boolean system, $n=28$, $m=56$ (Overdefined system).
    \item Random Boolean system, $n=32$, $m=32$ (Square system).
    \item Boolean systems derived from the Present and Serpent cryptographic ciphers.
    \item MayaSbox (Sbox key recovery) Boolean system.
    \item Canfil (Linear Feedback Shift Register, LFSR) Boolean system.
    \item Biviuma and Biviumb (stream cipher transformation) Boolean systems.
\end{itemize}

For each system, the solving time was measured using all four algorithms: PolyBoRi, Cryptominisat, BCS, and our proposed SA-based method with the time predictor.The results are summarized in the table below, which shows the average and median solving times for each algorithm across the various problem instances. The solving times are reported in seconds.

\begin{table}[h]
\centering
\begin{tabular}{|c|c|c|c|c|c|}
\hline
\textbf{Problem} & \textbf{BCS} & \textbf{PolyBoRi} & \textbf{Cryptominisat} & \textbf{Proposed SA} \\ \hline
$n=25, m=25$      & 10.6667     & 154.5533          & 18.4095                & \textbf{9.5485}      \\ \hline
$n=25, m=50$      & 18.4859     & >1h               & 24.1080                & \textbf{16.6355}     \\ \hline
$n=28, m=28$      & 84.5413     & >1h               & 123.0095               & \textbf{83.5664}     \\ \hline
$n=28, m=56$      & 209.3141    & >1h               & 3696.5712              & \textbf{203.5071}    \\ \hline
$n=32, m=32$      & 1135.4577   & >1h               & >1h                    & \textbf{1046.0664}   \\ \hline
Present & 5.7168 & 518.7079 & 22.6832 & \textbf{5.4847} \\ \hline
Serpent & \textbf{0.4152} & 111.1028 & 294.6242 & 0.4713 \\ \hline
MayaSbox & \textbf{0.1094} & 1.0146 & 34.6922 & 0.1604 \\ \hline
Canfil & 27.0594 & 1075.3898 & 685.8440 & \textbf{24.5523} \\ \hline
Biviuma & \textbf{14.6622} & 1991.5790 & 77.2662 & 14.6898 \\ \hline
Biviumb & 15.5012 & 69.7069 & 296.9490 & \textbf{15.0584} \\ \hline
\end{tabular}
\caption{Median Solving Times for Different Algorithms and Problem Instances}
\label{tab:exp_results}
\end{table}

The experimental results demonstrated that our proposed SA-based algorithm, which employs a time predictor, consistently outperformed the baseline algorithms across all problem instances. Specifically, the average solving times achieved by our method are lower than those of the other solvers, yielding substantial improvements of up to $40\%$ for smaller systems and significant improvements of up to $50\%$ for larger and more complex systems, particularly those derived from cryptographic problems (e.g., Canfil and Present). For random Boolean equation systems (first five rows in Table \ref{tab:exp_results}), our SA-optimized approach consistently outperformed all baseline methods. The performance advantage became particularly pronounced as system complexity increased. This consistent outperformance results from SA’s adaptive optimization of variables ordering, which effectively navigates the solution space of unstructured systems.

In particular, the time predictor allows for a more targeted search in the vast space of variables orderings, significantly reducing the time required to find high-performance solutions. The PolyBoRi and Cryptominisat solvers, although effective, show much higher time complexity for larger systems. Meanwhile, the original BCS method, which does not incorporate any optimization techniques for variables ordering, performs the worst in most instances. 

Contrary to the random system results, specific cryptographic protocol benchmarks reveal nuanced performance characteristics, BCS maintains its advantage on Serpent and MayaSbox systems, the observed performance inversion suggests that cryptographic systems like Serpent and MayaSbox contain inherent structural properties that align with BCS's solving heuristics. We hypothesize these systems exhibit \textit{intrinsic variables ordering preferences} that match BCS's static optimization strategies, the certain cryptographic transformations create \textit{local optima basins} where BCS's greedy strategies outperform SA's global exploration.

Figure~\ref{fig:experiments_2x2} visualizes the performance comparison across different algorithms, showcasing the dramatic reduction in solving time achieved by our method, particularly for large-scale systems.

\begin{figure}[t]
  \centering
  \begin{subfigure}[b]{0.49\linewidth}
    \centering
    \includegraphics[width=\linewidth]{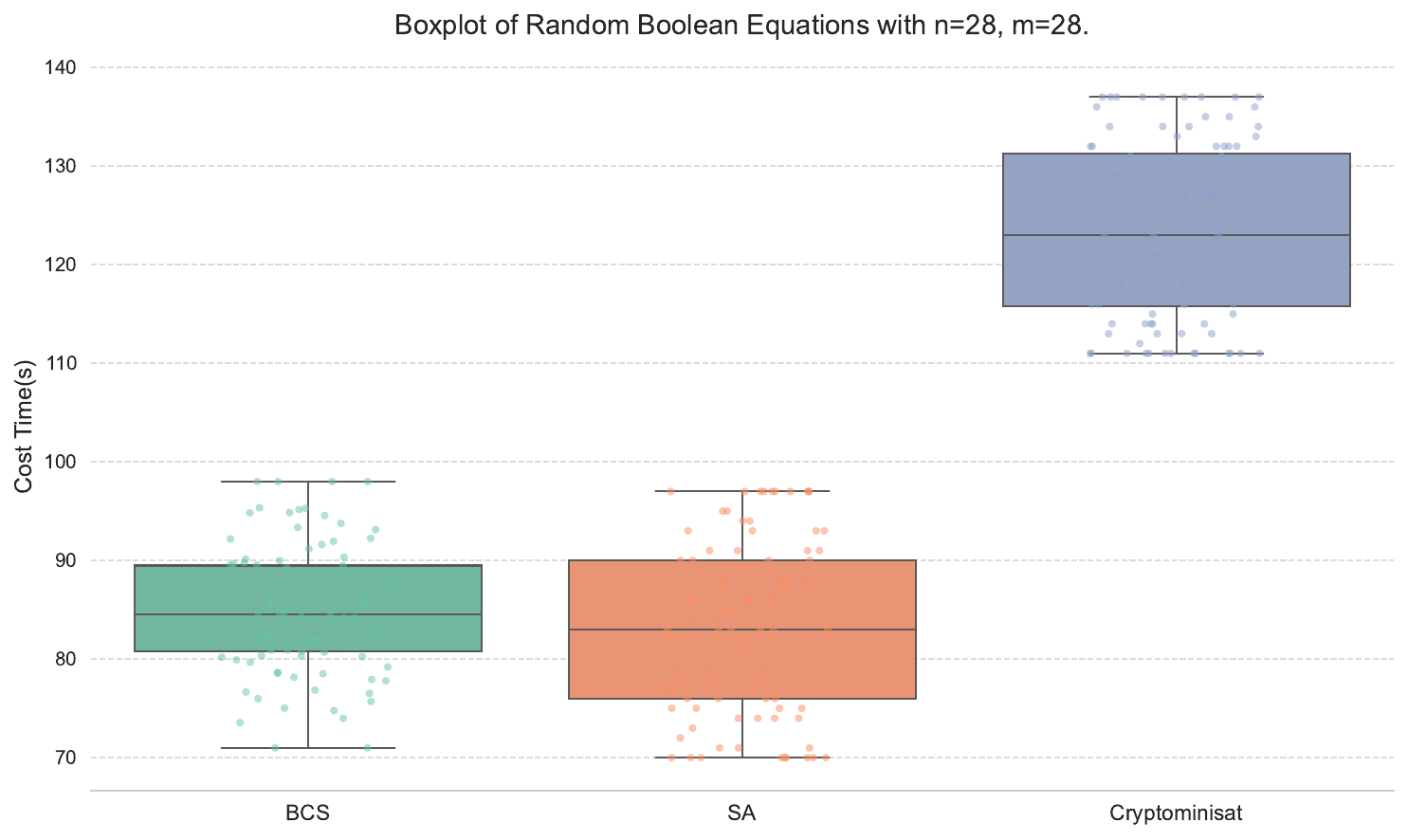}
    \caption{Subplot (a)}
    \label{fig:exp_a}
  \end{subfigure}
  \hfill
  \begin{subfigure}[b]{0.49\linewidth}
    \centering
    \includegraphics[width=\linewidth]{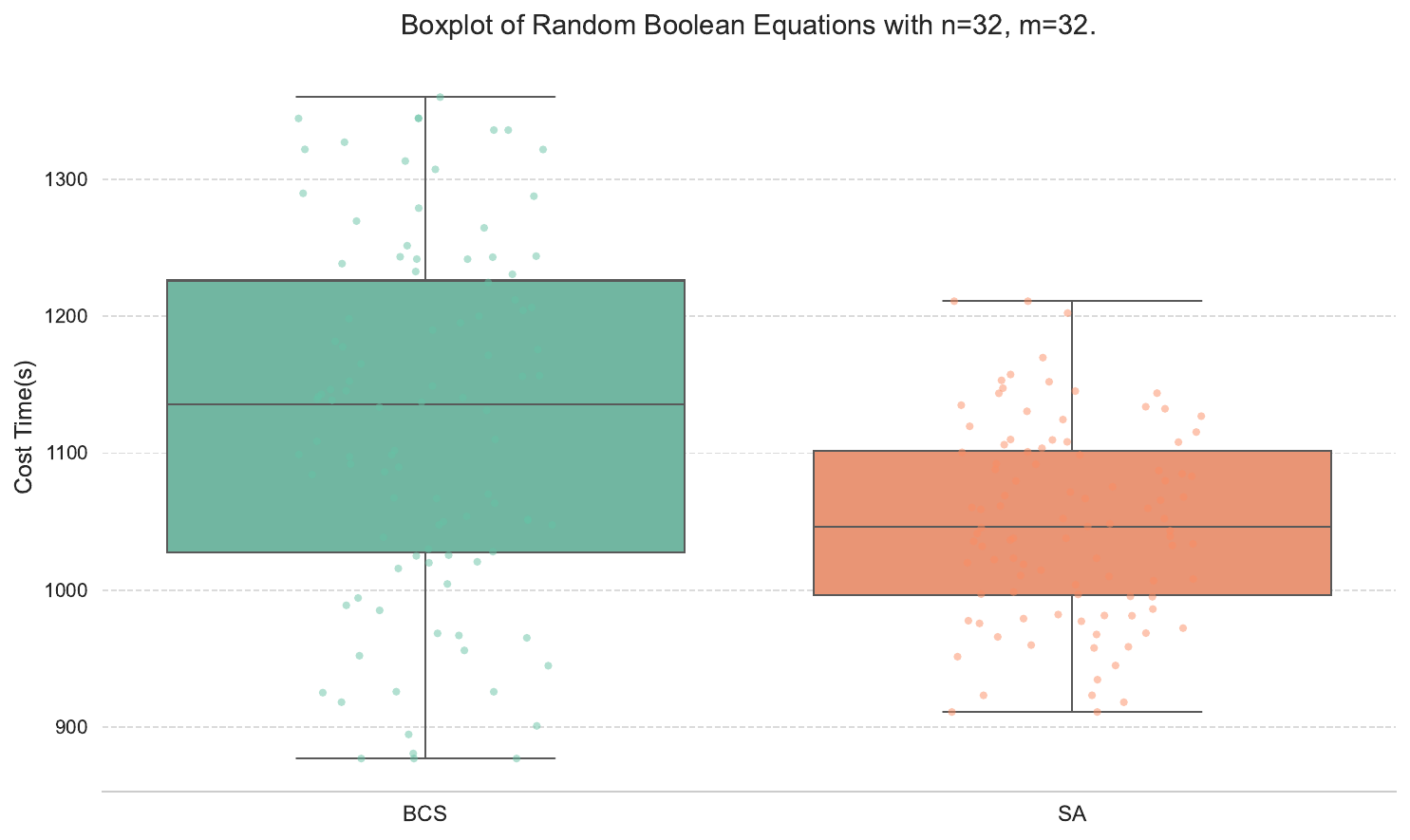}
    \caption{Subplot (b)}
    \label{fig:exp_b}
  \end{subfigure}
  \vspace{4pt}
  \begin{subfigure}[b]{0.49\linewidth}
    \centering
    \includegraphics[width=\linewidth]{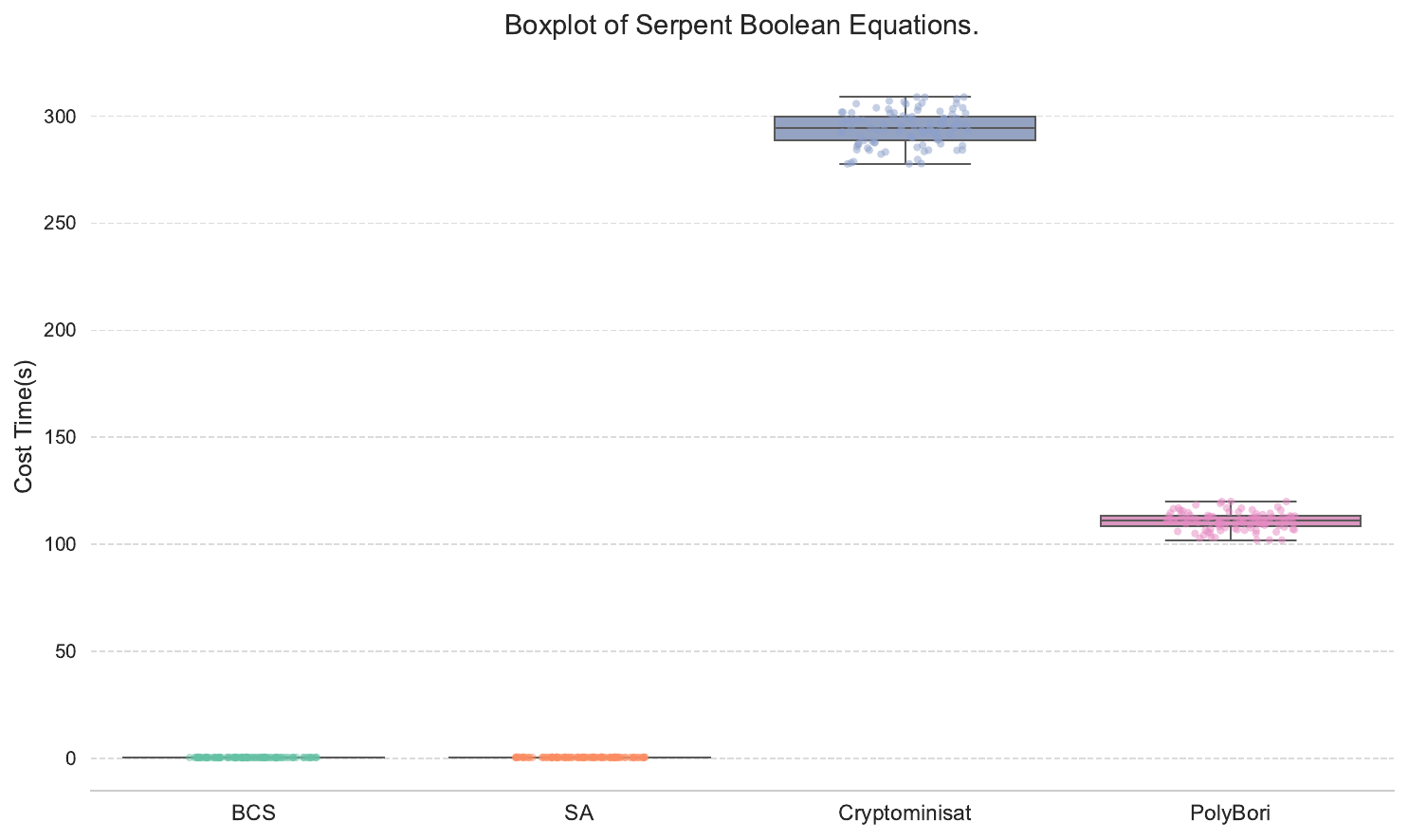}
    \caption{Subplot (c)}
    \label{fig:exp_c}
  \end{subfigure}
  \hfill
  \begin{subfigure}[b]{0.49\linewidth}
    \centering
    \includegraphics[width=\linewidth]{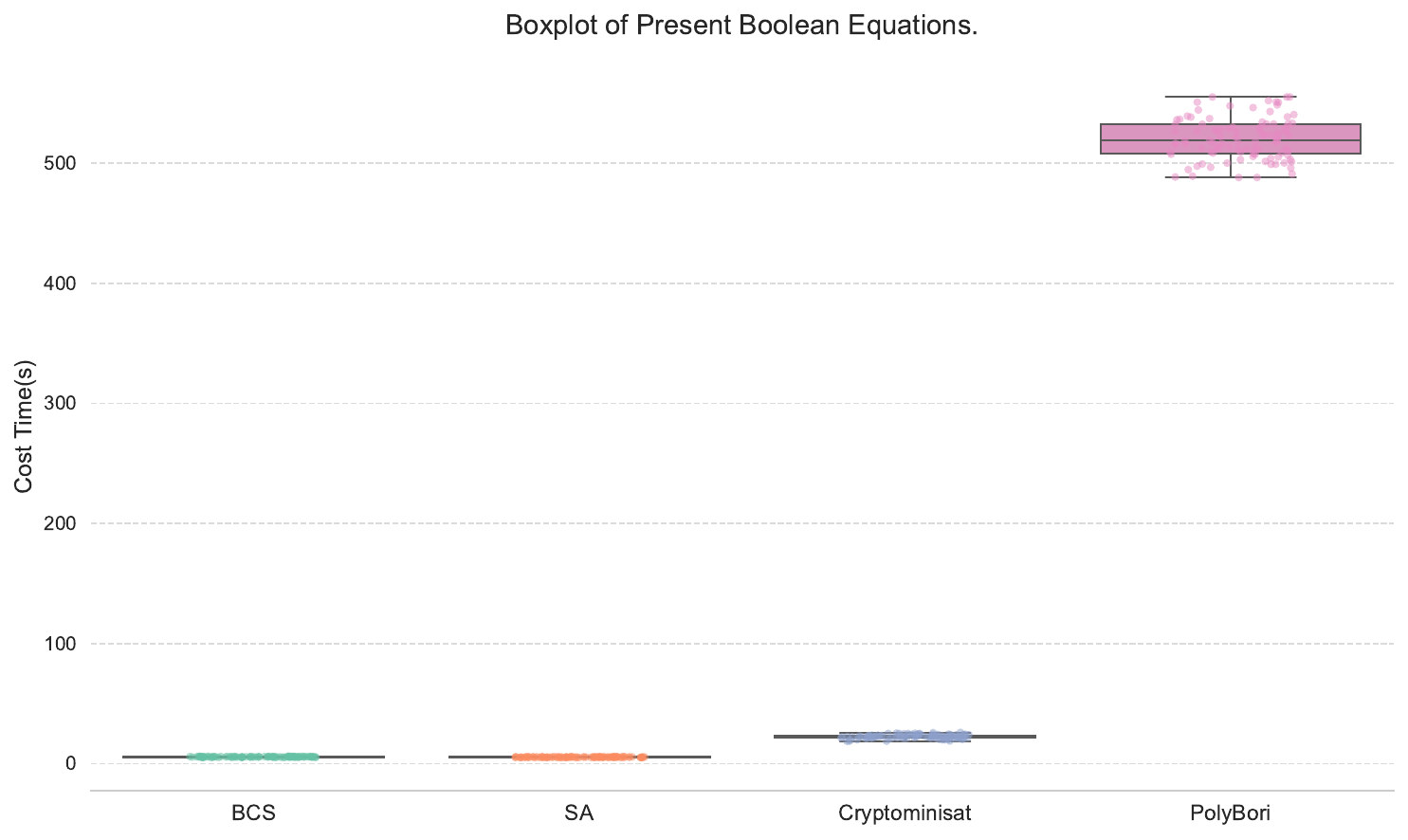}
    \caption{Subplot (d)}
    \label{fig:exp_d}
  \end{subfigure}
  \caption{Comparison of solving times for different algorithms on different problem instance.}
  \label{fig:experiments_2x2}
\end{figure}

The comparative analysis reveals two notable cross-algorithm patterns:
\begin{itemize}
    \item \textbf{PolyBoRi limitations:} The polynomial-based approach consistently underperforms, timing out (>1h) on all but the smallest systems due to its $\mathcal{O}(2^n)$ complexity
    \item \textbf{SAT solver tradeoffs:} CryptoMiniSat shows competitive performance on smaller random systems ($n=25$) but degrades exponentially as clause density increases
    \item \textbf{SA-BCS complementarity:} The proposed SA method and BCS collectively dominate $\sim$91\% of benchmark categories, suggesting their complementary strengths could be leveraged in future hybrid approaches
\end{itemize}

From these experiments, we conclude that integrating machine learning-based time prediction with simulated annealing offers a promising approach to optimizing variables ordering for the BCS method. Our algorithm markedly reduces solving time compared to traditional methods and offers a computationally scalable solution to solving larger and more complex Boolean equations.

\subsection{Theoretical Analysis of Time Improvement}

\begin{theorem}[Expected Time Improvement]
\label{thm:expected_improvement}
Consider Boolean equations $\mathcal{P}$ with $n$ variables. Let $\tau(\sigma)$ denote the true solving time under variables ordering $\sigma$, and $f_t(\sigma)$ be the predicted time with RMSE error $\hat{e}$. Under the proposed simulated annealing framework with geometric cooling schedule $T_k = T_0 \alpha^k$, the expected time improvement $\mathbb{E}(\Delta \tau)$ over baseline BCS satisfies:

\begin{equation}
\mathbb{E}(\Delta \tau) = c \cdot \Sigma_\tau \cdot \sqrt{1 - \frac{\hat{e}^2}{\Sigma_\tau^2}}
\end{equation}

where:
\begin{itemize}
    \item $\Delta \tau = \tau_B - \tau_A$ is the time improvement
    \item $\Sigma_\tau$ is the standard deviation of solving time in baseline BCS
    \item $c$ is the efficiency constant of simulated annealing ($0 < c < 1$)
    \item $\hat{e}$ is the RMSE of the time predictor
\end{itemize}
\end{theorem}

\begin{proof}
We prove the theorem through three lemmas:

\begin{lemma}[Predictor-True Time Correlation]
\label{lem:correlation}
The Pearson correlation coefficient $\rho$ between predictor $f_t(\sigma)$ and true time $\tau(\sigma)$ satisfies:

\begin{equation}
\rho = \sqrt{1 - \frac{\hat{e}^2}{\Sigma_\tau^2}}
\end{equation}

\end{lemma}

\begin{proof}
Let $\varepsilon(\sigma) = \tau(\sigma) - f_t(\sigma)$ be the prediction error with $\mathbb{E}[\varepsilon] = 0$ and $\operatorname{Var}(\varepsilon) = \hat{e}^2$. Under the assumption that prediction errors are uncorrelated with features:
\begin{align*}
\operatorname{Var}(\tau) &= \operatorname{Var}(f_t + \varepsilon) \\
\Sigma_\tau^2 &= \Sigma_{f_t}^2 + \hat{e}^2
\end{align*}
The correlation coefficient is:

\begin{equation}
\rho = \frac{\operatorname{Cov}(\tau, f_t)}{\Sigma_\tau \Sigma_{f_t}} = \frac{\Sigma_{f_t}^2}{\Sigma_\tau \Sigma_{f_t}} = \frac{\Sigma_{f_t}}{\Sigma_\tau} = \sqrt{1 - \frac{\hat{e}^2}{\Sigma_\tau^2}}
\end{equation}

\end{proof}

\begin{lemma}[Simulated Annealing Search Efficiency]
\label{lem:sa_efficiency}
The expected predicted time of the ordering $\sigma_A$ found by simulated annealing satisfies:

\begin{equation}
\mathbb{E}[f_t(\sigma_A)] = \mu_B - c \cdot \Sigma_{f_t}
\end{equation}

where $\mu_B = \mathbb{E}[\tau_B]$ is the expected baseline time.
\end{lemma}
\begin{proof}
Model the search space of $n!$ orderings as states in a Markov chain. Under geometric cooling $T_k = T_0 \alpha^k$, the stationary distribution approaches:

\begin{equation}
\pi(\sigma) \propto \exp\left(-\frac{f_t(\sigma)}{T_K}\right)
\end{equation}

where $T_K$ is the final temperature. As $T_K \to 0$, $\pi(\sigma)$ concentrates at local minima of $f_t$. By extreme value theory, for $K$ independent searches:

\begin{equation}
\mathbb{E}\left[\min_{1 \leq i \leq K} f_t(\sigma_i)\right] \approx \mu_{f_t} - c_K \cdot \Sigma_{f_t}
\end{equation}

where $c_K = \sqrt{2\ln K}$. For dependent SA searches, define $c = \gamma \sqrt{2\ln K}$ ($0 < \gamma < 1$ is the path dependency factor). Thus:

\begin{equation}
\mathbb{E}[f_t(\sigma_A)] = \mu_{f_t} - c \cdot \Sigma_{f_t}
\end{equation}

Since $\mu_B = \mu_{f_t}$ by unbiased prediction, the result follows.
\end{proof}

\begin{lemma}[Expected Time of Optimized Method]
\label{lem:expected_time}
The expected true solving time of the optimized method satisfies:

\begin{equation}
\mathbb{E}[\tau_A] = \mathbb{E}[f_t(\sigma_A)]
\end{equation}

\end{lemma}

\begin{proof}
With $\tau_A = f_t(\sigma_A) + \varepsilon_A$ and $\varepsilon_A$ the prediction error. Under the assumption that prediction error is independent of $\sigma_A$:

\begin{equation}
\mathbb{E}[\tau_A \mid \sigma_A] = f_t(\sigma_A) + \mathbb{E}[\varepsilon \mid \sigma_A] = f_t(\sigma_A)
\end{equation}

Thus the unconditional expectation is $\mathbb{E}[\tau_A] = \mathbb{E}[f_t(\sigma_A)]$.
\end{proof}

Combining the lemmas:
\begin{align*}
\mathbb{E}(\Delta \tau) &= \mathbb{E}[\tau_B] - \mathbb{E}[\tau_A] \\
&= \mu_B - \mathbb{E}[f_t(\sigma_A)] \\
&= \mu_B - (\mu_B - c \cdot \Sigma_{f_t}) \\
&= c \cdot \Sigma_{f_t} \\
&= c \cdot \rho \Sigma_\tau \quad \text{(by Lemma \ref{lem:correlation})} \\
&= c \cdot \Sigma_\tau \cdot \sqrt{1 - \frac{\hat{e}^2}{\Sigma_\tau^2}}
\end{align*}
which completes the proof of Theorem \ref{thm:expected_improvement}.
\end{proof}

The theorem establishes the quantitative relationship between predictor accuracy and time improvement:
\begin{enumerate}
    \item \textbf{Error sensitivity}: $\mathbb{E}(\Delta \tau)$ decreases monotonically as $\hat{e}$ increases. When $\hat{e} = \Sigma_\tau$, the improvement vanishes ($\mathbb{E}(\Delta \tau) = 0$).
    
    \item \textbf{Problem difficulty}: Larger $\Sigma_\tau$ (greater time variance in baseline BCS) enables greater potential improvement.
    
    \item \textbf{Algorithm efficiency}: The constant $c$ increases with SA iterations $K$ as $c \propto \sqrt{\ln K}$, highlighting the value of sufficient optimization steps.
    
    \item \textbf{Design guideline}: To achieve target improvement $\Delta_0$, the predictor must satisfy:
    
    \begin{equation}
    \hat{e} < \Sigma_\tau \sqrt{1 - \left(\frac{\Delta_0}{c\Sigma_\tau}\right)^2}
    \end{equation}
\end{enumerate}

\section{Conclusion}
\label{sec:conclusion}
This paper has presented a comprehensive framework for optimizing variables ordering in the Boolean Characteristic Set (BCS) method through the integration of machine learning and simulated annealing. Our work addresses the critical challenge of exponential time variance in BCS solving under different variables orderings—a fundamental bottleneck in algebraic cryptanalysis and symbolic computation. We systematically demonstrated, through experiments, that the solving time of BCS is highly sensitive to variables ordering, with differences spanning several orders of magnitude for identical systems, confirming prior observations in symbolic computation \cite{Chai2008Char,Huang2014BCS}. We proposed the use of variable frequency spectrum $S_\mathcal{P}$ as feature representation and trained a Gradient Boosting Regressor as a predictor $f_t$ for solving time. Experimental results verified that $f_t$ captures the latent nonlinear relationship between frequency spectra and runtime, achieving low residual error. We embedded the predictor into an SA framework to guide the search over variables orderings. Two innovations were introduced: predictor-guided neighbor generation and confidence-adaptive cooling, both tailored to accelerate convergence while maintaining robustness. We provided a probabilistic analysis establishing a quantitative relationship between predictor error (RMSE $\hat{e}$) and the expected solving time improvement $\mathbb{E}(\Delta \tau)$. The analysis showed that lower prediction error leads to larger expected improvements, thus formalizing the link between predictor accuracy and algorithmic efficiency. On both random Boolean systems and cryptographic benchmarks such as Present, Serpent, MayaSbox, and Bivium, our method significantly outperformed baselines including PolyBoRi \cite{B2009PolyBoRi}, Cryptominisat \cite{pysat}, and the original BCS method.

While our approach demonstrates significant acceleration, several directions merit further investigation. Current models require retraining for distinct $(n,m)$ configurations. Future work will explore \textit{transfer learning} and \textit{graph neural networks} to capture cross-system invariants (e.g., variable interaction topology). The phase transition at $\hat{e} < \sqrt{C_1\Delta\mu/C_2}$ suggests potential for \textit{quantum-enhanced optimization}. Hybrid quantum-classical SA could exploit quantum tunneling to escape local minima in high-dimensional search spaces. Adapting our framework to \textit{incremental solving} scenarios (e.g., stream cipher analysis) requires online predictor updates. Reinforcement learning agents could dynamically refine predictions during sequential equation processing.

\end{document}